\documentclass[11pt]{article}

\usepackage{fullpage}

\usepackage[numbers]{natbib}
\usepackage{dsfont}
\usepackage{latexsym}
\usepackage{amsmath}
\usepackage{amsthm}
\usepackage{amssymb}
\usepackage{amsfonts}
\usepackage{mathrsfs}
\usepackage{aliascnt}
\usepackage{graphicx}
\usepackage[pdfpagelabels,pdfpagemode=None]{hyperref}
\usepackage{algorithmic}
\usepackage{algorithm}

\newtheorem{theorem}{Theorem}%

\newaliascnt{lemma}{theorem}
\newtheorem{lemma}[lemma]{Lemma}%
\aliascntresetthe{lemma}

\newtheorem{corollary}{Corollary}%

\newtheorem{proposition}{Proposition}%

\theoremstyle{definition}
\newtheorem{definition}{Definition}

\newtheorem{example}{Example}

\newtheorem{remark}{Remark}

\newcommand{\AutoAdjust}[3]{\mathchoice{ \left #1 #2  \right #3}{#1 #2 #3}{#1 #2 #3}{#1 #2 #3} }
\newcommand{\Xcomment}[1]{{}}

\newcommand{\InBrackets}[1]{\AutoAdjust{[}{#1}{]}}
\newcommand{\Ex}[2][]{\operatorname{\mathbf E}_{#1}\InBrackets{#2}}

\newcommand{\given}{\;\mid\;}
\newcommand{\eps}{\epsilon}

\newcommand{\noaccents}[1]{#1}

\newcommand{\newagentvar}[3][\noaccents]{%
\expandafter\newcommand\expandafter{\csname #2\endcsname}{#1{#3}}%
\expandafter\newcommand\expandafter{\csname #2s\endcsname}{#1{\boldsymbol{#3}}}%
\expandafter\newcommand\expandafter{\csname #2smi\endcsname}[1][i]{#1{\boldsymbol{#3}}_{-##1}}%
\expandafter\newcommand\expandafter{\csname #2i\endcsname}[1][i]{#1{#3}_{##1}}%
\expandafter\newcommand\expandafter{\csname #2ith\endcsname}[1][i]{#1{#3}_{(##1)}}%
}

\newcommand{\ndist}{m}
\newcommand{\dist}{D}
\newcommand{\distfam}{\mathcal{D}}
\newcommand{\distj}[1][j]{\dist^{#1}}
\newagentvar{val}{v}
\newagentvar{typespace}{T}
\newagentvar{util}{u}
\newagentvar{distmat}{\dist}
\newagentvar{alloc}{x}
\newagentvar{payment}{p}

\newcommand{\utilvcg}{\util^{\mathrm{SPA}}}
\newcommand{\nsample}{k}
\newcommand{\sample}{\mathbf s}
\newcommand{\lottery}{L}
\newcommand{\conddist}{F}

\newcommand{\dister}{\vec d}

\newcommand{\distdim}{d}
\newcommand{\base}{B}
\newcommand{\decomp}{\alpha}
\newcommand{\basekron}{C}
\newcommand{\veronese}{\nu}

\newagentvar{composemat}{H}

\title{Optimal Auctions for Correlated Buyers with Sampling \thanks{N.\@ Haghpanah and J.\@ Hartline are supported by NSF CAREER Award CCF-0846113 and NSF CCF-1101717.  N.\@ Haghpanah is also supported by a Simons Graduate Award 285006.  R.\@ Kleinberg acknowledges support from NSF award AF-0910940, a Microsoft Research New Faculty Fellowship, and a Google Research Grant.}}
\author{
Hu Fu \\
Microsoft Research \\
New England Lab\\
{\tt hufu@microsoft.com}
\and
Nima Haghpanah \\
Northwestern University \\
Department of Eletrical Engineering \\
and Computer Science \\
\tt{nima.haghpanah@gmail.com}
\and 
Jason Hartline \\
 Northwestern University \\
Department of Electrical Engineering \\
and Computer Science \\
\tt{hartline@eecs.northwestern.edu}
\and
Robert Kleinberg \\
Cornell University \\
Department of Computer Science \\
\tt{rdk@cs.cornell.edu}
}

\begin{document}

\begin{titlepage}

\maketitle

\begin{abstract}

Cr\'emer and McLean [1985] showed that, when buyers' valuations are drawn from a
correlated distribution, an auction with full knowledge on the distribution can
extract the full social surplus.   We study whether this phenomenon persists
when the auctioneer has only incomplete knowledge of the distribution,
represented by a finite family of candidate distributions, and has sample access
to the real distribution.  We show that the naive approach which uses samples to
distinguish candidate distributions may fail, whereas an extended version of the
Cr\'emer-McLean auction simultaneously extracts full social surplus under each
candidate distribution.  With an algebraic argument, we give a tight bound on
the number of samples needed by this auction, which is the difference between
the number of candidate distributions and the dimension of the linear space they span.

\end{abstract}

\thispagestyle{empty}
\end{titlepage}

\section{Introduction}
\label{sec:intro}
Revenue maximization, a.k.a.\@ optimal auction design, is one of the most
studied topic in the literature of mechanism design.  The foundational work of
\citet{M81} gave the optimal auction for selling a single item when the
participating agents' values for the item are each drawn independently from
distributions that are known to the auctioneer.  As noted by \citeauthor{M81}
himself, the independence of the values is crucial for his auction's
optimality.   \citet{CM85} showed that, for correlated buyers, with a mild
condition on the joint distribution, there is in general a dominant strategy incentive
compatible auction that extracts full social surplus.  This means that,  the item is always sold to the bidder with the highest value, but the
utility of each bidder is zero in expectation, and all the value created by
the auction is extracted by the payments made to the auctioneer.

Both \citeauthor{M81}'s and \citeauthor{CM85}'s auctions are subject to the
criticism of the ``Wilson's principle'' \citep{W87}, which proposes that 
auction design should rely as little as possible on the details of the valuation
distributions, since these distributional prior knowledge can be expensive, if
not impossible, to acquire.  Recently, several \emph{prior-independent} auctions
were studied \citep[e.g.][]{DRY10, RTY12,
FHH13}.  These auctions greatly reduce the amount or precision of the knowledge
needed by the auctioneer on the prior distribution, while guaranteeing nearly
optimal performance against the fully knowledgeable optimal auction.  However,
they were all designed for independent value settings.  In this work, we pursue
prior-independent auction design for correlated buyers.

Typically, a prior-independent auction is an auction that is guaranteed to have
nearly optimal performance on a family of distributions.  To be more precise,
for each distribution of a given family, the auction extracts nearly as much
revenue as the optimal auction designed specifically for that distribution.  For
example, \citet{DRY10} showed that, with one sample from the underlying
distribution for each bidder, a single auction can extract a revenue nearly
optimal, and the family of distributions considered is that of regular distributions.

Can one use a similar method to emulate \citet{CM85}'s auction for correlated
settings, perhaps using samples to weaken the requirement on one's precise
knowledge on the value distributions?  We give an affirmative answer to this
question.

\begin{theorem}
\label{thm:internal}
Given any finite family~$\distfam$ of distributions satisfying the
condition of \citet{CM85}, a single auction with sample access to the underlying
distribution can extract full social surplus under each distribution
in~$\distfam$, where the
number of samples needed is one plus $|\distfam|$ (the number of
distributions in~$\distfam$),
minus the dimension of the linear space spanned by the distributions
in~$\distfam$.  
\end{theorem}

For example, when all distributions in the family are linearly independent, the
auction needs only one sample.  We also show that the number of samples needed
is tight (\autoref{prop:lb-sample}).

\paragraph{Main Difficulties.}

Unlike the independent value settings, for correlated bidders, the optimal
auctions that guarantee bidders non-negative utilities only in expectation
(i.e., interim individually rational) may extract a revenue much higher than any
auction which guarantee non-negative utilities with probability one (i.e.,
ex post individually rational).  In general, the social surplus can be extracted
only with the former looser constraint on the auction.
However, the expectation of a bidder's utility here is taken over the
distribution of the other bidders' values, on which the auctioneer has no accurate knowledge.  Therefore, the disadvantage imposed
by the inaccuracy of the auctioneer's knowledge is considerably greater here
than in the independent value settings.  We allow ourselves a finite family of
distributions to partly compensate for this disadvantage.

Since we consider a finite family of distributions, a naive way to make use of
the sample access
is trying to distinguish the distributions by what one sees in the samples, and
then running the auction optimal for the distribution which most likely produces the samples.  In \autoref{sec:external}, we show that this approach can be very
inefficient in its use of the samples.

\begin{theorem}
\label{thm:external}
For any positive integer~$\nsample$, there exist two distributions for two
bidders which satisfy
the Cr\'emer-McLean condition and which, with high probabilities, cannot be
distinguished by $\nsample$~samples.  Moreover, no single auction extracts any constant fraction of the optimal revenue under both distributions.
\end{theorem}

In contrast, the auction in \autoref{thm:internal} needs only one sample for
two such distributions.

Moreover, even when the auctioneer has some confidence on his guess of the
underlying distribution and runs an auction that is interim invidually rational
for that distribution, the auction may run into trouble in the event that 
his guess is wrong (such an event is particularly likely if the only source of
confidence is the samples).  This is because the
auction may turn out to be not invidiaully rational on the actual distribution,
and participants may see negative utilities in expectation.
In this case, one will have to make additional assumptions on bidders' behavior
in such scenarios for any meaningful analysis.  Our approach avoids this
problem: as long as the actual distribution is in the family~$\distfam$,
 our auction will be individually rational.

\paragraph{Our Techniques.}

Our auction is an extension of the auction by \citet{CM85} (called the
CM-auction in the sequel).  That auction first runs a second price auction on
the reported values, and then charges each bidder a payment (or pays her a reward) which is
determined solely by the other bidders' bids.  Since this payment does not
depend on the bidder's own strategy, it does not alter the incentive structure
of the auction.  These payments can be seen as the outcome of a lottery, whose
randomness comes only from the other bidders' values.  The lotteries 
are set up so that each bidder, conditional on any of her own values, makes an expected
payment in the lottery that is equal to her expected utility in the second
price auction.  In the independent value settings, this is not possible, because
the outcome of the lottery does not depend on the bidder's own value, whereas
her expected utility in the second price auction does.  In the case of
correlated values, this becomes possible if there is enough ``richness'' or
``variance'' in the conditional distributions of the other bidders' values as
the bidder's own value varies.  This ``richness'' is shown by \citeauthor{CM85}
to be the linear independence of the conditional distributions.


In our auction, we also first run a second price, and then decide for each bidder the
outcome of her lottery without using her own value.  The difference of our
auction from the
CM-auction is that this lottery outcome will depend not only on the other bidders'
values but also on the samples from the distribution.  Ideally, even though the
bidder's expected utility in the second price auction changes with both her own
value \emph{and} the underlying distribution, we hope to orchestrate the change
in the expected lottery outcome so as to match the utility change.  This requires
linear independence of all the distributions over the other bidders' values and
the samples, given each candidate distribution and the bidder's own value.
The main technical
difficulty of this work is to give a tight bound on the number of samples needed
for this linear independence.  As we show in the proof of
\autoref{thm:internal}, it boils down to showing a property for an object in
algebraic geometry known as the Veronese variety.

\paragraph{Structure of the paper.}

In \autoref{sec:external} we show the limit of the naive approach by giving a
proof for \autoref{thm:external}.  In \autoref{sec:internal} we prove
\autoref{thm:internal}, first describing our auction and then showing the bound on the number of samples for its revenue guarantee.

\subsection{Other Related Works}
\label{sec:related}

In correlated value settings with a known distribution, \citet{CM85} gives a
dominant strategy incentive compatible, interim individually rational auction
that extracts full social surplus, under a certain condition
(\autoref{def:CM-cond}) on the distribution.  Our work is an extension of this
auction.  The CM auction was extended by \citet{MR92} and \citet{Rahman10} to
continuous type spaces.


Another line of work studies the optimal auction for the same setting (with
known prior distributions), but under the stronger constraint of ex post
individual rationality.  \citet{PP11} showed that calculating the optimal
deterministic auction under this requirement is NP-hard, whereas \citet{DFK11}
showed that the optimal randomized auction can be computed in time polynomial in
the size of the distribution.  \citet{R01} developed a $2$-approximation for
the optimal revenue where the computation cost does not grow with the number of
bidders, and this approach was extended by \citet{DFK11} and \citet{CHLW11} for
better approximation ratios.  These auctions are particularly simple in form,
and we will use the auction by \citeauthor{R01} in our proof of
\autoref{thm:external}.  For the more general matroid settings, \citet{RT13}
characterized the optimal auction under various assumptions on the distribution, and
\citet{Li13} showed that a generalized VCG auction with conditional monopoly
reserve prices gives $e$-approximation to the optimal revenue for distributions having a correlated version of monotone hazard rate.


There have been various studies on prior-independent revenue maximization
\citep[e.g.][]{DRY10,DHKN11,RTY12,FHH13}, although they all assume independent value distributions.  The most relevant to this work is
\citet{DRY10}'s single-sampling auction, which showed that with one sample from
each bidder's valuation distribution, the VCG auction with the samples as
reserve prices gives a $4$-approximation to the optimal revenue, when the
distributions are regular.  As an extention, \citet{RT13} gave a single-sampling mechanism for the more general interdependent value settings under various assumptions, although the benchmark is the optimal revenue under ex post individual rationality.  Recently, \citet{CHMS13} gave a prior-independent
mechanism optimizing a non-revenue objective, i.e., that of minimizing makespan for scheduling problems.


Online pricing \citep[e.g.\@][]{BBDS11, BDKS12, BKS13} is another setting where one has to maximize revenue but faces an
unknown underlying distribution, and where one can observe values drawn from it
(or partial information, e.g.\@ by observing the buyer's decision to take or
leave a certain price).  The difference between this and our ``batch''
setting is that the observations come over time, and one needs to
perform well not only in the last stage, but throughout the stages on average.  Also,
the buyers are typically assumed to have values (or types) drawn from the same
distribution, as opposed to from a correlated distribution we consider
here.

\section{Preliminaries}
\label{sec:prelim}
\paragraph{Auctions, Incentive Compatibility and Invidual Rationality.}

In this paper we consider the problem of auctioning one item to $n$ bidders
whose private valuations $\vali[1], \ldots, \vali[n]$ are drawn from an unknown
correlated distribution $\dist$.  Let $\typespacei$ be the support of
$\vali$, i.e., $\typespacei = \{\vali \given \exists \vali[-i], \dist(\vali,
\vali[-i]) > 0\}$.  Let $\typespace$ be the support of~$\dist$.  In this work we
consider only discrete distributions with finite supports.

By the revelation principle, it is without
loss of generality to consider auctions of the form that solicit bidders'
values and map them to an \emph{allocation} and a \emph{payment} for each
bidder.  The allocation $\alloci (\vali[1], \ldots, \vali[n]) \in [0, 1]$
denotes the probability with which agent~$i$ is allocated the item at the
reported value profile $(\vali[1], \ldots, \vali[n])$, and the payment
$\paymenti (\vali[1], \ldots, \vali[n])$ indicates the amount of money paid by
agent~$i$ at the valuation profile.  Feasibility of a single-item auction
requires that $\sum_{i} \alloci (\vali[1], \ldots, \vali[n]) \leq 1$, $\forall
\vali[1], \ldots, \vali[n]$.

The \emph{utility}~$\util_i$ of a bidder with value~$\vali$ at an outcome $\alloci$ and
$\paymenti$ is $\vali \alloci - \paymenti$.  An auction is said to be
\emph{dominant strategy incentive compatible (DSIC)}, if for all $i$, $\vali$, $\vali'$ and $\vali[-i]$, 
\begin{align*}
\vali \alloci(\vali, \vali[-i]) - \paymenti(\vali, \vali[-i]) \geq
\vali \alloci(\vali', \vali[-i]) - \paymenti(\vali', \vali[-i]).
\end{align*}

An auction is said to be \emph{ex post individually rational (IR)} if, for all
$i$, $\vali$, $\vali[-i]$,
\begin{align*}
\vali \alloci(\vali, \vali[-i]) - \paymenti(\vali, \vali[-i]) \geq 0.
\end{align*}
An auction is said to be \emph{interim individually rational} if, for all $i$
and~$\vali$,
\begin{align*}
\Ex[\val_{-i}]{\vali \alloci(\vali, \vali[-i]) - \paymenti(\vali, \vali[-i]) \given \vali} \geq 0,
\end{align*}
where $\vali[-i]$ is drawn from the conditional distribution given
$\vali$.

The \emph{revenue} of an auction is the sum of expected payments it collects
from all bidders.  In this paper we consider maximizing revenue extractable by
auctions subject to dominant strategy IC condition and ex post or interim IR condition.

\paragraph{Notations for Distributions.}

We will assume that the auctioneer is guaranteed that the valuation
distribution~$\dist$ is from a family~$\distfam$ of distributions, $\distfam
= \{\distj[1], \cdots, \distj[\ndist]\}$.  
\footnote{Without loss of generality we assume these distributions have the same support.}
For example, the auctioneer may have an accurate knowledge of the
distribution~$\dist$ on the value profiles $(\vali[1], \ldots, \vali[n])$, but
without knowing the mapping between the identities of the bidders in the auction
and the coordinates in the valuation profile.  In this case, the family
$\distfam$ will consist of at most $n!$ distributions, each of which is formed by performing a permutation on the coordinates in the profiles in~$\dist$.

There are multiple ways to represent a distribution.  In \autoref{sec:external}
we represent a distribution for $n$~bidders as an $n$-dimensional tensor.  In
particular, for two bidders, a distribution~$\dist$ is a $|\typespacei[1]|
\times |\typespacei[2]|$ matrix, with the entry $\dist(\vali[1], \vali[2])$ denoting the probability of the occurrence of $(\vali[1], \vali[2])$.  In
\autoref{sec:internal}, we will represent a distribution by a
$|\typespace|$-dimensional vector, with $\dist(\vali[1], \ldots, \vali[n])$
 being the probability of the occurrence of valuation profile $(\vali[1],
\ldots, \vali[n])$.  When $\typespace = \prod_{i} \typespacei$, the latter is
simply the vectorization of the former representation.

The probability of a valuation profile~$\vali[-i]$ conditioning on bidder~$i$'s
value being~$\vali$ is $\dist(\vali[-i] \given \vali)$.  We use $\dist_{i,
\vali}$ to denote the conditional distribution on $\vali[-i]$ given $\vali$.  We
represent it as the $|\typespacei[-i]|$-dimensional vector, where $\dist_{i, \vali, \vali[-i]}$ is $\dist(\vali[-i] \given \vali)$.

\paragraph{Optimal Auctions For A Known Distribution.} 

We will need two existing results on revenue maximization with correlated bidders, under constraints of interim IR and ex post IR, respectively.

\citet{CM85} showed that, under a fairly lenient condition on the value
distribution, the optimal mechanism under DSIC and interim IR can extract the full
social surplus.  In other words, the auction maximizes the social welfare and
always allocates the item to the bidder with the highest value, whereas in
expectation every bidder's utility is zero.

\begin{definition}
\label{def:CM-cond}
A valuation distribution~$\dist$ is said to satisfy the \emph{Cr\'emer-McLean
condition} if, for each bidder~$i$, the $|\typespacei|$ vectors $\{\dist_{i,
\vali}\}_{\vali \in \typespacei}$ are linearly independent.
\end{definition}

\begin{theorem}[\citealp{CM85}]
\label{thm:CM}
In a single item auction where the valuation distribution satisfies the
Cr\'emer-McLean condition, there is an interim IR, DSIC auction that
extracts the full social surplus.
\end{theorem}

We will call the auction in \autoref{thm:CM} the CM auction.

\citet{R01} studied an DSIC, ex post IR \emph{lookahead} auction that
$2$-approximates the optimal revenue.  The auction first solicits all values,
then singles out the highest bidder, and runs the optimal auction for this
bidder, with the value distribution conditioning on all other bidders' values
and the fact that her value is above all others'.  

\begin{theorem}[\citealp{R01}]
\label{thm:lookahead}
The lookahead auction is DSIC, ex post IR and extracts at least half of
the optimal revenue.  
\end{theorem}

\paragraph{The Equal Revenue Distribution.}

In several examples we will make use of the following \emph{equal revenue}
distribution (truncated at~$h$): the valuation $\val$ takes on integers between
$1$ and~$h$, and the probability that $\val \geq k$ is equal to $\tfrac 1 k$.
The equal revenue distribution has the property that, the expectation of the
value is $\Omega(\log h)$, which grows unboundedly as $h$ grows large, but
the optimal revenue one can extract from it in a single-agent setting is $1$.

\paragraph{Kronecker Products.} 

Notations in \autoref{sec:internal} will be greatly shortened by the use of Kronecker
products on matrices (and vectors).  The Kronecker product of matrices $A =
(a_{ij}) \in \mathbb R^{m \times n}$ and $B = (b_{ij}) \in \mathbb R^{p \times
q}$ is the $mp \times nq$ block matrix
\begin{align*}
A \otimes B = \left[
\begin{array}{ccc}
a_{11} B & \cdots & a_{1n} B \\
\vdots  & \ddots & \vdots \\
a_{m1} B & \cdots & a_{mn} B
\end{array}
\right].
\end{align*}

Kronecker products are bilinear and associative, but in general are not
commutative.  We will use $(\otimes A)^k$ to denote the Krocker product of $k$
copies of~$A$.  When performing Kronecker products on an $m$-dimensional vector,
we will treat it as an $m \times 1$ matrix.  The following lemma is not hard to
verify.

\begin{lemma}
\label{lem:kronecker-lin-ind}
Consider a set of linearly independent vectors $S = \{v_1, \ldots v_m\}$
and, for each $i = 1, \ldots, m$, a set $T_i$ of
linearly independent vectors.  The set of vectors $\cup_{i = 1, \cdots, m} \{u
\otimes v_i\}_{u \in T_i}$ are linearly independent.  In particular, for any positive
integer~$k$, the set of $m^k$ vectors $\{u_1 \otimes
\ldots \otimes u_k\}_{u_i \in S}$, are linearly independent.
\end{lemma}


\section{The Limit of the Naive Approach}
\label{sec:external}

The most naive approach given sample accesses is to use the samples to
distinguish distributions in the given family~$\distfam$ of distributions, and
then run an optimal auction for the identified distribution.  However, even with
the knowledge of~$\distfam$, the auctioneer may still need a large number of
samples to even distinguish the distributions with constant confidence of being
right, let alone tailor an auction for the identified distribution.  Consider the following example.

\begin{example}
\label{ex:coin}
Fix a small positive real number~$\eps < 1$.
Consider two bidders whose values are generated by the following process: two
values are independently drawn from the equal revenue distribution, then with
probability $1 - \eps$, the two values are randomly assigned to the two bidders;
with probability $\eps$, the higher of the two values is assigned to bidder~$1$,
and the lower to bidder~$2$.  Call the resulting correlated
distribution~$\distj[A]$.  Define another distribution~$\distj[B]$ by exactly the same procedure but flipping the identity of the two bidders.
\end{example}

\begin{proposition}
\label{prop:coin}
It takes $\Omega(1/\eps^2)$ samples to correctly distinguish $\distj[A]$ and
$\distj[B]$ in Example~\ref{ex:coin} with constant probability.
\end{proposition}

\begin{proof}
We can simulate a biased coin with $\distj[A]$ as follows: draw a pair of values
$(\vali[1], \vali[2])$ from~$\distj[A]$, and if $\vali[1] > \vali[2]$, return
Head; if $\vali[1] < \vali[2]$, return Tail; if $\vali[1] = \vali[2]$, return
Head and Tail with probability $\tfrac 1 2$ each.  It is not hard to see that
the resulting distribution over Heads and Tails is that of an $\eps$-biased coin
in favor of Head.  The same simulation using $\distj[B]$ will give a
distribution of an $\eps$-biased coin in favor of Tail.  By standard information
theoretic argument \citep[see, e.g. Theorem~6.1 in][]{KK07}, we know that it takes $\Omega(\eps^2)$ flips of a coin to
distinguish an $\eps$-biased coin in favor of Heads or Tails.  Therefore one
needs at least as many samples to distinguish $\distj[A]$ and~$\distj[B]$.
\end{proof}

It is not hard to verify that $\distj[A]$ and~$\distj[B]$ satisfy the Cr\'emer-McLean condition,
and one can extract full social surplus which is $\Omega(\log h)$.  Now we
show that, without being able to distinguish the two, one auction cannot
simultaneously be interim IR and approximates the optimal revenue within a
constant factor under both distributions.

\begin{theorem}
\label{thm:lb-external}
There is no auction that is interim IR and gets more than $O(1)$ revenue under
both distributions in Example~\ref{ex:coin} .
\end{theorem}

Together with \autoref{prop:coin}, this theorem implies \autoref{thm:external}.
Before proving the theorem, we first give a characterization of dominant
strategy IC
mechanisms, which can be easily shown by standard arguments.  We omit its proof.

\begin{lemma}
\label{lem:ic-char}
Given a value distribution, any two dominant strategy IC auctions with the same
allocation rule differ from each other only by a payment from each bidder~$i$
that depends only on~$\vali[-i]$.
\footnote{We allow the payment to be negative, in which case we pay the bidder.}
\end{lemma}

 For a fixed allocation rule, we call an auction
\emph{canonical} if it has the allocation rule, is dominant strategy IC, ex post IR and if
any bidder having a value lowest in her type space always has zero utility.
Given Lemma~\ref{lem:ic-char}, we can describe any dominant strategy IC auction
by the difference between it and the canonical one.

\begin{corollary}
Any dominant strategy IC auction can be fully described in a standard form by its
allocation rule and $n$ vectors $\lottery_1 \in \mathbb R^{|\typespacei[-1]|}$, 
$\cdots, \lottery_n \in \mathbb R^{|\typespacei[-n]|}$.  To run this auction,
one first runs the canonical auction with the same allocation rule, and then
charges each bidder~$i$ the amount $\lottery_i(\vali[-i])$ when the other bidders bid $\vali[-i]$.
\end{corollary}

\begin{proof}[Proof of \autoref{thm:lb-external}]
Fix a dominant strategy IC auction that is interim IR under both $\distj[A]$
and~$\distj[B]$ in Example~\ref{ex:coin}.  Let $\lottery_1 \in \mathbb
R^{|\typespacei[2]|}$ and~$\lottery_2 \in \mathbb R^{|\typespacei[1]|}$ be the
vectors of payments describing the auction's payments in addition to the
canonical auction.  Let $\dister$ denote the equal revenue distribution
truncated at~$h$, i.e., $\dister = [\tfrac 1 2, \tfrac 1 6, \cdots,
\tfrac{1}{(h-1)h}, \tfrac{1}{h}]^\top$.  Deviating from the rest of the paper, in
this proof we will use a $h\times h$ matrix~$\dist$ to
represent a joint distribution for the two bidders, with $\dist(\vali[1],
\vali[2])$ representing the probability of the profile $(\vali[1], \vali[2])$.
The distribution where each bidder's value is drawn independently from the equal
revenue distribution is then represented by the symmetric rank-one matrix $A = \dister
\dister^\top$.  Let $B$ be the upper triangle matrix whose diagonal elements are
half of $A$'s, and whose elements above the diagonal are identical with~$A$.
Then $A = B + B^\top$.  Also, $\distj[A]$ is $(1 - \eps) A + 2 \eps B^\top = A + \eps(B^\top - B)$, and $\distj[B]$ is $A +
\eps (B - B^\top)$.  

We consider the revenue our auction gets from bidder~$1$ under~$\distj[A]$.  It
first gets the revenue as in the canonical auction, and then in addition, it
gets $\vec 1^\top \cdot [A + \eps (B^\top - B)] \lottery_1$, where $\vec 1^\top$
is the $h$-dimensional all-one row vector.  Under
$\distj[B]$, this additional revenue is $\vec 1^\top \cdot [A + \eps (B - B^\top)]
\lottery_1$.  The sum of these two terms is $\vec 1^\top \cdot A \lottery_1$,
which we show has to be small.  

Let $r_1(M)$ denote the first row of a matrix~$M$.  Recall that the entries of $r_1(A + \eps
(B^\top - B))$ correspond to the probabilities in~$\distj[A]$ of the value profiles where
$\vali[1] = 1$.  By interim IR, we have 
\begin{align*}
r_1(A + \eps (B^\top - B)) \cdot \lottery_1 = r_1(A) \lottery_1 + r_1(B^\top -
B) \cdot \lottery_1 \leq 1.
\end{align*}
Similarly, for $\distj[B]$ we have
\begin{align*}
r_1(A) \lottery_1 + r_1(B - B^\top) \cdot \lottery_1 \leq 1.
\end{align*}
Adding the two inequalities, we get $r_1(A) \lottery_1 \leq 2$.  Now recall that $A$ represents the independent distribution, and $r_1(A)$ is simply $\tfrac 1 2 \dister^\top$.
Hence $\dister^\top \cdot \lottery_1 \leq 4$.  The other rows of~$A$ are also scaled copies of $\dister$,
and the scaling coefficients sum up to~$1$.  Hence, $\vec 1^\top
A \lottery_1 \leq 4$.  Therefore, in addition to the canonical auction, the
total sum of our auction's extra revenue from 
bidder~$1$ in the two distributions is bounded by~$4$.  The same argument works
for bidder~$2$ as well.  In other words, our auction cannot extract
substantially more revenue than the canonical auction for both distributions
simultaneously.  Therefore,
to finish the proof, we only need to show that the canonical auction also
extracts only a small revenue.  

To show this, we invoke the lookahead auction.  By \autoref{thm:lookahead}, the
revenue of any DSIC, ex post IR auction, including that of the canonical
auction, is bounded by twice the revenue of the lookahead auction.  Recall that
the lookahead auction for two bidders uses the lower bidder's value to set a
conditionally optimal price for the higher bidder.  Given any distribution
represented by a $h \times h$ matrix~$\dist$, when bidder~$1$'s value
being~$\vali[1]$ and bidder~$2$'s value is higher, the optimal price for
bidder~$2$ is determined by the part of $\vali[1]$-th row of~$\dist$ to the
right of the diagonal element.  If we set a price of $\payment_{1, \vali[1]}
\geq \vali[1]$ in
this scenario, in expectation we collect a revenue of 
\begin{align*}
\payment_{1, \vali[1]} \sum_{\val \geq \payment_{1, \vali[1]}} \dist(\vali[1], \val).
\end{align*}
We can do this for every~$\vali[1]$, and symmetrically for bidder~$2$ using the
columns of $\dist$.  The revenue of the lookahead auction can be expressed as
\begin{align*}
\max_{
\substack{\payment_{1, 1} \geq 1, \\
\cdots \\
\payment_{1, h} \geq h.
}
}
\sum_{\vali[1] = 1}^h \payment_{1, \vali[1]} \sum_{\val \geq \payment_{1,
\vali[1]}} \dist(\vali[1], \val) +
\max_{
\substack{\payment_{2, 1} > 1, \\
\cdots \\
\payment_{2, h - 1} > h - 1.
}
}
 \sum_{\vali[2] = 1}^h \payment_{2, \vali[2]}
\sum_{\val \geq \payment_{2, \vali[2]}} \dist(\val, \vali[2]).
\end{align*}
Now we substitute $\dist$ by $\distj[A] = A + \eps(B^\top - B)$.   It is not
hard to verify that the diagonal
elements of $\distj[A]$ are the same as~$A$, and for any $\vali[1] >
\vali[2]$, $\distj[A](\vali[1], \vali[2]) = (1 + \eps) A(\vali[1], \vali[2])$,
and for any $\vali[1] < \vali[2]$, $\distj[A](\vali[1], \vali[2]) = (1 - \eps)
A(\vali[1], \vali[2])$.  Therefore, the revenue of the lookahead auction is
upper bounded by $(1 + \eps)$ times the following quantity:
\begin{align*}
\max_{
\substack{\payment_{1, 1} \geq 1, \\
\cdots \\
\payment_{1, h} \geq h.
}
}
\sum_{\vali[1] = 1}^h \payment_{1, \vali[1]} \sum_{\val \geq \payment_{1,
\vali[1]}}  A(\vali[1], \val) +
\max_{
\substack{\payment_{2, 1} > 1, \\
\cdots \\
\payment_{2, h - 1} > h - 1.
}
}
 \sum_{\vali[2] = 1}^h \payment_{2, \vali[2]}
\sum_{\val \geq \payment_{2, \vali[2]}}  A(\val, \vali[2]).
\end{align*}

This quantity, however, is exactly the revenue of the lookahead auction for the
independent distribution~$A$, which in turn is known to be no more than~$2$
(since the revenue of Myerson's optimal auction is no more than~$2$; another way
to see this is that for independent distributions, the optimal revenue for two
bidders cannot be greater than the sum of optimal revenue extractable from each
alone).  This completes the proof.
\end{proof}

\section{The Power of Auctions with Internal Samples}
\label{sec:internal}

In contrast to the limit we have shown for the naive approach, we consider an
extension of the CM-auction in this section and prove \autoref{thm:internal}.

\begin{definition}
\label{def:cm-sample}
The \emph{CM auction with samples} works as follows:
\begin{enumerate}
\item Run the second price auction, which allocates the item to the highest bidder and
charges her a payment equal to the second highest bid.
\item Draw $\nsample$ valuation profiles $\sample_1, \ldots, \sample_\nsample$,
each independently from the underlying distribution.
\item \label{step:lottery} For each bidder, including those who do not win the
item, charge her or pay her an amount of money that is a function of the other
bidders bids $\vali[-i]$ \emph{and} the samples $\sample_1, \ldots,
\sample_\nsample$.  These functions and $\nsample$, the number of samples
needed, are to be specified later.
\end{enumerate} 
\end{definition}

The difference between the CM auction with samples and the CM auction is the sampling procedure and the dependence of the lottery outcome on
the samples.  We now discuss setting up the lotteries outcomes in
Step~\eqref{step:lottery}, and the number of samples we need.  The former is an
extension of the CM auction, whereas the latter involves nontrivial algebraic investigations.

\subsection{Lottery Outcomes From Solving Linear Systems}
\label{sec:lottery}

The construction of the lotteries in Step~\eqref{step:lottery} of
\autoref{def:cm-sample} aims at extracting from bidder~$i$ the utility she would
get in a pure second price auction, no matter what distribution we are under.  This boils
down to solving a linear system, as is the case in \citet{CM85}.

For each bidder~$i$, we construct a vector $\utilvcg_i$ in $\mathbb
R^{|\typespacei| \times \ndist}$, where $\utilvcg_{i, \vali, j}$ is the expected
utility of bidder~$i$ in the second price auction under distribution $\distj$ and
conditioning on that bidder~$i$ has value~$\vali$.  (Recall that $\ndist$ is the
number of distributions in~$\distfam$.)  We draw $\nsample$~samples
$\sample_1, \ldots, \sample_\nsample$ from the underlying distribution, where
each sample $\sample_j$ is a profile of values $(\sample_{j1}, \ldots,
\sample_{jn})$.  We would like to decide on an amount to pay or charge
bidder~$i$ given $\vali[-i]$ and $\sample_1, \ldots, \sample_\nsample$.  So we
use a vector $\lottery_i \in \mathbb R^{|\typespacei[-i]| \times
|\typespace|^\nsample}$ to denote these quantities, where $\lottery_{i,
\vali[-i], \sample_1, \ldots, \sample_\nsample}$ is the amount of money we
charge or pay to bidder~$i$, when the other bidders bid~$\vali[-i]$ and when the
samples are $\sample_1, \ldots, \sample_\nsample$.  To compute the expected
payment under~$\lottery$, we need a distribution over the events that
$\vali[-i], \sample_1, \ldots, \sample_\nsample$ occur.  Importantly, this
distribution varies with both the underlying distribution~$\distj$ \emph{and}
bidder~$i$'s own value~$\vali$.  Therefore we have $|\typespacei| \cdot \ndist$
vectors $\conddist^{\vali, j}_i$ in $\mathbb R^{|\typespacei[-i]| \times
|\typespace|^\nsample}$, where $\conddist^{\vali, j}_{i, \vali[-i], \sample_1,
\ldots, \sample_{\nsample}}$ is the probability that the bidders other than~$i$
bid $\vali[-i]$ and that the $\nsample$~samples are $\sample_1, \ldots,
\sample_\nsample$, under the joint distribution~$\distj$ \emph{and} conditioning
on bidder~$i$'s own value being~$\vali$.  The expected payment that bidder~$i$
with value~$\vali$ makes in Step~\eqref{step:lottery} under
distribution~$\distj$  is then equal to $\conddist^{\vali, j}_i \cdot
\lottery_i$.

\begin{proposition}
\label{prop:cm-icir}
The CM auction with samples is DSIC.  In addition, if, given a
family of distributions $\distfam = \{\distj[1], \cdots, \distj[\ndist]\}$, for each bidder~$i$, the system of linear equations 
\begin{align}
\conddist^{\vali, j}_i \cdot \lottery_i = \utilvcg_i, \quad \forall \vali \in
\typespacei, j \in [\ndist]
\end{align}
has a solution $\lottery_i^*$, then using $\lottery^*_i$ for bidder~$i$ in
Step~\eqref{step:lottery} of \autoref{def:cm-sample} makes the auction interim
IR and extracts full social surplus under each distribution~$\distj \in
\distfam$.
\end{proposition}

\begin{proof}
The second price auction itself is DSIC, and in Step~\eqref{step:lottery} of
\autoref{def:cm-sample}, the extra payment (or award) the bidder makes (or
receives) is not affected by her own bid, the auction remains DSIC.

Now fixing any distribution $\distj \in \distfam$, and conditioning on
bidder~$i$ having value~$\vali$, the bidder's utility from the first two steps
will be her conditional utility in a second price auction, i.e., $\utilvcg_{i, \vali,
j}$.  Her extra payment in Step~\eqref{step:lottery} will be in expectation
equal to $\conddist^{\vali, j}_i \cdot \lottery_i^*$, which by definition of
$\lottery^*_i$ is equal to $\utilvcg_{i, \vali, j}$.  This shows that the bidder
has expected utility zero no matter which $\distj \in \distfam$ it is and no
matter what her own value is.  Therefore the auction is interim IR.  As the item
is always allocated to the highest bidder, the auction extracts the full social
surplus.
\end{proof}

We now investigate conditions that allow us to solve the linear systems
$\conddist^{\vali, j}_i \cdot \lottery_i = \utilvcg_i$.  From this point on we
will focus on the problem on a fixed bidder, and will drop the subscripts~$i$.  
In general, there are no linear constraints governing the entries of the vector
$\utilvcg$, because it is calculated with both the probabilities in the
distribution \emph{and} the magnitude of the valuations.  This means
that, to be able to solve the linear equations, in general we need
$\{\conddist^{\vali, j}\}_{\vali, j}$ to be linearly independent.

By the independence of each sampling, we have 
\begin{align*}
\conddist^{\vali, j}_{\vali[-i], \sample_1, \ldots, \sample_\nsample} =
\distj_{\vali}(\vali[-i]) \cdot \distj(\sample_1) \cdots
\distj(\sample_\nsample),
\end{align*}
 (recall that $\distj_{\vali}(\vali[-i])$ is the
conditional probaiblity $\distj(\vali[-i] \given \vali)$), therefore
\begin{align}
\conddist^{\vali, j} = \distj_{\vali} \otimes (\otimes \distj)^{\nsample}.
\label{eq:conddist}
\end{align}  
By the bilinearity of Kronecker products, in
order to have $\{\conddist^{\vali, j}\}_{\vali \in \typespacei}$ to be linearly
independent even for a fixed~$j$, we need $\{\distj_{\vali}\}_{\vali}$ to
be linearly independent, which amounts to the Cr\'emer-McLean condition
(\autoref{def:CM-cond}) on~$\distj$.

From this point on we will assume that each $\distj \in \distfam$ satisfies the
Cr\'emer-McLean condition, and we look at the number of samples needed to make
$\{\conddist^{\vali, j}\}_{\vali, j}$ linearly independent.

\subsection{Upper Bounds on the Number of Samples Needed}
\label{sec:ub-sample}

We next show the main theorem in this section.

\begin{theorem}
\label{thm:ub-sample}
If each distribution $\distj \in \distfam$ satisfies the Cr\'emer-McLean
condition, and if the $\ndist$ vectors $\{\distj\}$ spans a linear space of
dimension~$\distdim$, then with $\nsample = \ndist - \distdim + 1$ samples,  the
set of vectors $\{\conddist^{\vali, j}_i\}_{\vali, j}$ are linearly independent, for each bidder~$i$.
\end{theorem}

With Proposition~\ref{prop:cm-icir}, we immediately have the following
corollary.

\begin{corollary}
\label{cor:ub-sample}
Under the condition in \autoref{thm:ub-sample}, the CM auction with $\nsample =
\ndist - \distdim + 1 $ 
samples is DSIC, interim IR, and extracts full social surplus under each
distribution $\distj \in \distfam$.
\end{corollary}

\begin{proof}[Proof of \autoref{thm:ub-sample}]
By \eqref{eq:conddist} and Lemma~\ref{lem:kronecker-lin-ind}, as we have the
Cr\'emer-McLean condition, it suffices to show that the $\ndist$~vectors $\{(\otimes \distj)^\nsample\}_{j}$ are linearly independent.

Let $\{\base_1, \cdots, \base_{\distdim}\}$ be a basis of the linear space
spanned by $\{\distj\}_j$.  Then for each~$j$, we can write $\distj$ as a linear
sum of these vectors: $\distj = \sum_{\ell = 1}^{\distdim} \decomp_{j \ell}
\base_{\ell}$.  Since each $\distj$ is a distribution, its entries sum to one.
Therefore, no two $\decomp_j$ and $\decomp_{j'}$ are scalar copies of each
other, i.e., there are no $j \neq j'$ such that $\decomp_{j\ell} = \zeta
\decomp_{j' \ell}$ for each $\ell$, for some~$\zeta$.

We consider the Kronecker product $(\otimes \distj)^\nsample$.  By bilinearity,
\begin{align*}
\left( \otimes \distj \right)^{\nsample} 
& = \left( \otimes \sum_{\ell = 1}^\distdim \decomp_{j\ell} \base_\ell \right)^{\nsample} \\
& = \sum_{\substack{\ell_1 + \dots + \ell_{\distdim} = \nsample, \\
\ell_1, \cdots, \ell_{\distdim} \geq 0}} \decomp_{j1}^{\ell_1}
\decomp_{j2}^{\ell_2} \ldots \decomp_{j\distdim}^{\ell_\distdim} 
\basekron_{\ell_1, \cdots, \ell_\distdim},
\end{align*}
where $\basekron_{\ell_1, \cdots, \ell_{\distdim}}$ is the sum of terms that are
Kronecker products of $\base_1, \cdots, \base_\distdim$, such that in each term $\base_1$ appears
$\ell_1$ times, and so on.  (since taking kronecker product is not commutative,
these products do not have to be the same.) for example, when $\distdim$ is two,
$\basekron_{1, 2} = \base_1 \otimes \base_2 \otimes \base_2 + \base_2 \otimes
\base_1 \otimes \base_2 + \base_2 \otimes \base_2 \otimes \base_1$.  by
lemma~\ref{lem:kronecker-lin-ind}, the set of vectors $\{\base_{\ell_{1}}
\otimes \cdots \otimes \base_{\ell_{\nsample}}\}_{\ell_1, \cdots, \ell_\nsample
\in [\distdim]}$ are linearly independent, and therefore so are the
$\basekron_{\ell_1, \ldots, \ell_\distdim}$'s.

now each $(\otimes \distj)^\nsample$ is expressed as a linear combination of
linearly independent vectors, with the linear coefficient on $\basekron_{\ell_1,
\cdots, \ell_\distdim}$ being the product $\decomp_{j1}^{\ell_1}\ldots
\decomp_{j\distdim}^{\ell_{\distdim}}$.  to show linear independence of the set of
vectors $\{(\otimes \distj)^\nsample\}_j$, we only need to show that the set of
$\ndist$ linear coefficients as vectors are linearly independent.

the vector $(\decomp_{j1}^{\ell_1} \ldots \decomp_{j \distdim}^{\ell_{\distdim}})_{\ell_1 +
\cdots + \ell_\distdim = \nsample}$ is the image of the vector
$\vec \decomp_j = (\decomp_{j1}, \ldots, \decomp_{j\distdim})$ under the mapping
$\veronese: \mathbb r^{\distdim} \to \mathbb r^{{\distdim + \nsample - 1\choose
\distdim - 1}}$ which evaluates all the $\nsample$-th degree monomials in
$\mathbb r[x_1, \ldots, x_\distdim]$ at a point in $\mathbb r^{\distdim}$. 
we now show that these $\ndist$ images $\veronese(\vec \decomp_1), \ldots,
\veronese(\vec \decomp_\ndist)$ are
linearly independent when $\nsample = \ndist - \distdim + 1$.

we will show that for every~$j$, there exists a
linear form on $\mathbb r^{{\distdim + \nsample - 1 \choose \distdim - 1}}$ that
vanishes at $\veronese(\vec \decomp_{j'})$ for all $j' \neq j$ and does not vanish at
$\veronese(\vec \decomp_j)$.  this will show that there cannot be any linear
dependence among the $\ndist$ points $\veronese(\vec \decomp_j)$.

since $\{\distj\}_j$ spans a linear space of dimension~$\distdim$, and since
$\{\base_1, \cdots, \base_\distdim\}$ is a basis of this space, the vectors
$\vec \decomp_1, \ldots, \vec \decomp_\ndist$ span a $\distdim$-dimensional
linear space.  without loss of generality, consider $\vec \decomp_1$, 
we can find $\distdim - 1$ other vectors that are
linearly independent with $\vec \decomp_1$.  therefore we can find a linear form
$f_1: (y_1, \ldots, y_\distdim) \mapsto \beta_1 y_1 + \cdots + \beta_\distdim
y_\distdim$ which vanishes at all these $\distdim - 1$ vectors but does not
vanish at $\vec \decomp_j$.  without loss of generality, let the remaining $\ndist -
\distdim$ vectors be $\vec \decomp_{\distdim + 1}, \ldots, \vec
\decomp_{\ndist}$.  for each $j' = \distdim+1, \cdots, \ndist$, since $\vec
\decomp_{j'}$ is not a scaled copy of $\vec \decomp_j$, we can find a linear form $f_{j'}$
such that $f_{j'}$ vanishes at $\vec \decomp_{j'}$ but does not vanish at
$\vec \decomp_j$.  Now consider the product of these $\ndist - \distdim + 1$ linear forms,
\begin{align*}
f = f_1 f_{\distdim + 1} \ldots f_{\ndist}.
\end{align*}
If we take $\nsample$ to be $\ndist - \distdim + 1$, $f$ itself is a linear form
on $\mathbb R^{{\distdim + \nsample - 1 \choose \distdim - 1} }$, and can be
evaluated at $\veronese(\vec \decomp_1), \ldots, \veronese(\vec
\decomp_\ndist)$, and 
\begin{align*}
f(\veronese(\vec \alpha)) = f_1(\vec \alpha) f_{\distdim+1}(\vec \alpha) \ldots
f_{\ndist}(\vec \alpha), \quad \forall \vec \alpha \in \mathbb R^{\distdim}.
\end{align*}

By construction, $f(\veronese(\vec \decomp_{j})) = 0$ for all $j \neq 1$ and
$f(\veronese(\vec \decomp_1)) \neq 0$.  Since the choice of~$\vec \decomp_1$ was
arbitrary, the construction works for arbitrary~$\vec \decomp_j$, and so $\veronese(\vec \decomp_1), \ldots,
\veronese(\vec \decomp_\ndist)$ are linearly independent for $\nsample = \ndist - \distdim + 1$.  This completes the proof.
\end{proof}

\begin{remark}
In the last part of the proof, since no two $\vec \decomp_j, \vec \decomp_{j'}$ are linear copies of each other, the
$\ndist$ vectors $\vec \decomp_1, \ldots, \vec \decomp_\ndist$ can be seen as points in
the projective space $\mathbb P^{\distdim - 1}$.  The mapping $\veronese_\nsample:
\mathbb P^{\distdim - 1} \to \mathbb P^{{\distdim + \nsample \choose \distdim} - 1}$
is known as the Veronese embedding, and its image the Veronese variety.  In the
special case when $\distdim$ is two, the fact that no $\nsample + 1$ points on
$\veronese_\nsample(\mathbb P^1)$ are linearly dependent can be somewhat more
directly shown by an application of the Vandermonde determinant.
\end{remark}

\subsection{A Worst Case Lower Bound on the Number of Samples Needed}
\label{sec:lb-sample}

We now show that the number of samples specified in \autoref{thm:ub-sample} is
tight.

\begin{proposition}
\label{prop:lb-sample}
For any $\ndist$ and $\distdim < \ndist$, there exist $\ndist$ distributions
$\{\distj\}_j$ spanning a $\distdim$-dimensional linear space, such that for any
$\nsample \leq \ndist - \distdim$, the set of vectors $\{\conddist^{\vali, j}
\}_{\vali, j}$ are not linearly independent, for at least one bidder~$i$.
\end{proposition}

\begin{proof}
We first show that there are $\distj$'s that make $\{(\otimes
\distj)^\nsample\}_j$ linearly dependent for any $\nsample \leq \ndist -
\distdim$.  First consider the case $\distdim = 2$.

Let $\base_1$ and $\base_2$ be two independent vectors in the span of
$\distj$'s.  Then each $\distj$ can be written as $\decomp_{j1} \base_1 +
\decomp_{j2} \base_2$.  Following a similar calculation as in the proof of
\autoref{thm:ub-sample}, we have
\begin{align*}
\left( \otimes \distj \right)^{\nsample} & = \left( \otimes (\decomp_{j1} \base_1 + \decomp_{j2} \base_2) \right)^{\nsample} \\
& = \sum_{\substack{\ell_1 + \ell_2 = \nsample, \\
\ell_1,  \ell_{2} \geq 0}} \decomp_{j1}^{\ell_1} \decomp_{j2}^{\ell_2}  
\basekron_{\ell_1 \ell_2},
\end{align*}
where $\basekron_{\ell_1, \cdots, \ell_{\distdim}}$ is the sum of terms that are
Kronecker products of $\base_1$ and~$\base_2$, such that in each term $\base_1$
appears $\ell_1$ times, and $\base_2$ appears $\ell_2$ times.  For example,
$\basekron_{1, 2} = \base_1 \otimes \base_2 \otimes \base_2 + \base_2 \otimes
\base_1 \otimes \base_2 + \base_2 \otimes \base_2 \otimes \base_1$.  But we have
just shown that all the $\ndist$ vectors $(\otimes \distj)^\nsample$ can be
written as a linear combinations of $\nsample + 1$ vectors, $\basekron_{0,
\nsample}, \basekron_{1, \nsample - 1}, \cdots, \basekron_{\sample, 0}$.
Therefore, For $\nsample < \ndist - 1$, the vectors $(\otimes \distj)^\nsample$
cannot be linearly independent.

Now by \eqref{eq:conddist}, as long as we can construct, for one $\vali$, such that the conditional distribution $\distj_{\vali}$ is the same for all~$j$, then the vectors $\conddist^{\vali, j}$ cannot be linearly independent as well.  
This is easy to do, since $\distj_{\vali}$ only concerns a proper subset of
coordinates of~$\distj$, and we have complete freedom to construct the rest of the distribution.

The general case $\distdim > 2$ is an easy generalization of the case of
$\distdim = 2$: given $\distdim$ linearly independent $\distj[1], \cdots, \distj[\distdim]$, we
can always let the remaining distributions be linear combinations of $\distj[1]$
and~$\distj[2]$, and repeat the calculation above.
\end{proof}

\section{Discussion}
\label{sec:discuss}
\paragraph{Criticism on the Auction of \citeauthor{CM85}.}

The surplus-extracting auction of Cr\'emer and McLean is often seen as a critique on the model of auction design for correlated agents.  The (arguably) counterintuitive phenomenon of surplus extraction is ``blamed'' on the unrealistic combination of several assumptions in the model: first, that the agents are risk neutral and only considers their expected linear utilities; second, that the auctioneer has exact knowledge on the underlying distribution for the agents' values; and third, that the agents themselves have the same exact knowledge.  The second assumption, and the auctioneer's heavy use of this prior knowledge, is seen as a violation of the desired Wilson's principle.  Our result suggests tha the precision of the auctioneer's prior knowledge may not be the main cause of the mechanism's anomalous performance --- this requirement can be weakened, as long as sampling from the underlying distribution is available, and the number of samples does not have to be large.  This suggests fine-tuning criticism on these auctions on the agents' precise prior knowledge and the interim individual rationality assumption.  

In general, in Bayesian mechanism design, assumptions such as ``who knows what'' are crucial modeling decisions.  Our approach via sample complexity may be useful in examining mechanisms' sensitivity to these assumptions and hence help with fine-tuning the modeling process.

\paragraph{Beyond Finiteness.}

Even though our approach involves inverting matrices whose entries are probabilities of atom events, there may be hope to extend the approach to infinite-support distributions, since there have been such extensions to \citeauthor{CM85}'s auction \citep[e.g.][]{MR92, Rahman10}.  This seems a prerequisite for possibly extending the approach further to infinite families of distributions.  We think it would be interesting to either show an impassable gap between infinite and finite families, or give conditions that makes surplus extraction possible with finitely many samples on infinite families.

\section{Acknowledgement}
The authors would like to thank Nick Gravin and Gjergji Zaimi for helpful
discussions.

\bibliographystyle{apalike}
\bibliography{bibs}

\end{document}